%% file: ITW2017.tex
\documentclass[conference]{IEEEtran}

\usepackage{graphicx}
\usepackage{amsfonts}
\usepackage{amssymb}
\usepackage{amsthm}
\usepackage{amsmath}
\usepackage{color}
\usepackage{mathtools}
\usepackage{subcaption}
\usepackage{caption}
\usepackage{bigdelim,multirow,comment}
\usepackage[final]{pdfpages}

\newtheorem{theorem}{Theorem}
\newtheorem{lemma}{Lemma}

\newtheorem{corollary}{Corollary}

\newtheorem{remark}{Remark}

\newenvironment{myproof}[1][\proofname]{\proof[#1]}{\endproof}

\DeclarePairedDelimiter{\ceil}{\lceil}{\rceil}
\DeclarePairedDelimiter\floor{\lfloor}{\rfloor}
\allowdisplaybreaks

\usepackage{afterpage}

\ifCLASSINFOpdf
\else
\fi

\newcommand{\R}{\mathbb{R}}

\newcommand{\cR}{{\cal R}}

\graphicspath{{../main_document/Figures/}} % Specifies the directory where pictures are stored
\DeclareGraphicsExtensions{.pdf,.jpeg,.png,.jpg}
\def\b{\mathbf}
\DeclareMathOperator*{\argmin}{arg\,min}
\DeclareMathOperator*{\argmax}{arg\,max}

\begin{document}

\title{The Necessity of Scheduling in Compute-and-Forward}

\author{\IEEEauthorblockN{Ori Shmuel}
\IEEEauthorblockA{Department of Communication\\ System Engineering\\
Ben-Gurion University of the Negev\\
Email: shmuelor@bgu.ac.il}
\and
\IEEEauthorblockN{Asaf Cohen}
\IEEEauthorblockA{Department of Communication\\ System Engineering\\
Ben-Gurion University of the Negev\\
Email: coasaf@bgu.ac.il}
\and
\IEEEauthorblockN{Omer Gurewitz}
\IEEEauthorblockA{Department of Communication\\ System Engineering\\
Ben-Gurion University of the Negev\\
Email: gurewitz@bgu.ac.il}}

% conference papers do not typically use \thanks and this command
% is locked out in conference mode. If really needed, such as for
% the acknowledgment of grants, issue a \IEEEoverridecommandlockouts
% after \documentclass

% for over three affiliations, or if they all won't fit within the width
% of the page, use this alternative format:
% 
%\author{\IEEEauthorblockN{Michael Shell\IEEEauthorrefmark{1},
%Homer Simpson\IEEEauthorrefmark{2},
%James Kirk\IEEEauthorrefmark{3}, 
%Montgomery Scott\IEEEauthorrefmark{3} and
%Eldon Tyrell\IEEEauthorrefmark{4}}
%\IEEEauthorblockA{\IEEEauthorrefmark{1}School of Electrical and Computer Engineering\\
%Georgia Institute of Technology,
%Atlanta, Georgia 30332--0250\\ Email: see http://www.michaelshell.org/contact.html}
%\IEEEauthorblockA{\IEEEauthorrefmark{2}Twentieth Century Fox, Springfield, USA\\
%Email: homer@thesimpsons.com}
%\IEEEauthorblockA{\IEEEauthorrefmark{3}Starfleet Academy, San Francisco, California 96678-2391\\
%Telephone: (800) 555--1212, Fax: (888) 555--1212}
%\IEEEauthorblockA{\IEEEauthorrefmark{4}Tyrell Inc., 123 Replicant Street, Los Angeles, California 90210--4321}}

% use for special paper notices
%\IEEEspecialpapernotice{(Invited Paper)}

% make the title area
\maketitle

\begin{abstract}
%\boldmath

Compute and Forward (CF) is a promising relaying scheme which, instead of decoding single messages or forwarding/amplifying information at the relay, decodes linear combinations of the simultaneously transmitted messages. The current literature includes several coding schemes and results on the degrees of freedom in CF, yet for systems with a fixed number of transmitters and receivers. It is unclear, however, how CF behaves at the limit of a large number of transmitters. 

In this paper, we investigate the performance of CF in that regime. Specifically, we show that as the number of transmitters grows, CF becomes degenerated, in the sense that a relay prefers to decode only one (strongest) user instead of any other linear combination of the transmitted codewords, treating the other users as noise. Moreover, the sum-rate tends to zero as well. This makes scheduling necessary in order to maintain the superior abilities CF provides. Indeed, under scheduling, we show that non-trivial linear combinations are chosen, and the sum-rate does not decay, even without state information at the transmitters and without interference alignment. 

\end{abstract}

% For peer review papers, you can put extra information on the cover
% page as needed:
% \ifCLASSOPTIONpeerreview
% \begin{center} \bfseries EDICS Category: 3-BBND \end{center}
% \fi
%
% For peerreview papers, this IEEEtran command inserts a page break and
% creates the second title. It will be ignored for other modes.
\IEEEpeerreviewmaketitle

%%%%%%%%%%%%%%%%%%%%%%%%%%%%%%%%%%%%%%%%%%%

\section{Introduction}

Compute and Forward (CF) \cite{nazer2011compute} is a coding scheme which enables receivers to decode linear combinations of transmitted messages, exploiting the broadcast nature of wireless relay networks. CF utilizes the shared medium and the fact that a receiver, which received multiple transmissions simultaneously, can treat them as a superposition of signals, and decode linear combinations of the received messages. Specifically, together with the use of lattice coding, the obtained signal, after decoding, can be considered as a linear combination of the transmitted messages. This is due to an important characteristic of lattice codes - every linear combination of codewords is a codeword itself.

However, since the wireless channel suffers from fading, the received signals are attenuated by real (and not integers) attenuations factors, hence the received linear combination is "noisy". The receiver (e.g., a relay) then seeks a set of integer coefficients, denoted by a vector $\b{a}$, to be as close as possible\footnote{One can define different criteria for the goodness of the approximation, for example, the minimum distance between the vectors elements.} to the true channel coefficients.

This problem was elegantly associated with \emph{Diophantine Approximation Theory} in \cite{niesen2012degrees}, and was compared to a similar problem, that of finding a co-linear vector for the true channel coefficients vector (between the receiver and the transmitters). In addition, the co-linear vector must be an integer valued vector, due to the fact that it should represent the coefficients of an integer linear combination of codewords. Based on this theory, if one wishes to find an integer vector $\textbf{a}$ that is close (in terms of co-linearity) to a real vector $\textbf{h}$, then one must increase $||\textbf{a}||$ in order to have a small approximation error between them. The increase in the norm value leads to a significant penalty in the achievable rate at the receiver and thus results in a tradeoff between the goodness of the approximation and the maximization of the rate.

The CF scheme was extended in many directions, such as MIMO CF \cite{zhan2009mimo}, linear receivers (Integer Forcing) \cite{zhan2014integer}  \cite{sakzad2013integer}, integration with interference alignment \cite{niesen2012degrees}, scheduling \cite{he2015collision} and more \cite{wei2012compute}, \cite{hong2013compute}. All the mentioned works considered a general setting, where the number of transmitters is a parameter for the system and \emph{all transmitters are active at all times}. That is, the receiver is able to decode a linear combination of signals from a large number of transmitters as long as the transmitters comply with the achievable rates at the receiver, and still promise, to some extent, an acceptable performance.

However, in this work, we show that the number of simultaneous transmitters is of great importance when the number of relays is fixed. In fact, this number cannot be considered solely as a parameter but as a restriction since, when it grows, the receiver will prefer to decode only the strongest user over all possible linear combinations. This will make the CF scheme degenerated, in the sense that the relay chooses a vector $\b{a}$ which is actually a unit vector (a line in the identity matrix), thus treating all other signals as noise. In other words, the linear combination chosen is trivial. Furthermore, we show that as the number of transmitters grows, the scheme's \emph{sum-rate} goes to zero as well. Thus, one is forced to use \emph{users scheduling} to maintain the superior abilities CF provide.  

We conclude this paper with an optimistic view, that user scheduling can improve the CF gain. We believe that this can be done by suitable matching of linear combinations, i.e. coding possibilities. Using simple Round Robin scheduling and results for CF in fixed size systems, we lower bound the sum-rate. We thus show that even for a simple scheduling policy the system sum-rate does not decay to zero.

%\subsection{Main contributions}

The paper is organized as follows. In Section \ref{sec-System model}, the system model is described. In Section \ref{sec-Probability of a unit vector}, we derive an analytical expression for the probability of choosing a unit vector by the relay, as the number of users grows. Section \ref{sec-Compute and Forward Sum-Rate} depicts the behaviour of the sum-rate for this model, and in Section \ref{sec-Scheduling in CF} we present the advantage of using scheduling, along with a simple scheduling algorithm.

%%%%%%%%%%%%%%%%%%%%%%%%%%%%%%%%%%%%%%%%%%%%%%%%%%%%%%%%%%%%%%%%%%%%%%%%%%%%%%%%%%%%%%

\begin{figure}[t]
\centering
    \includegraphics[width=0.4\textwidth]{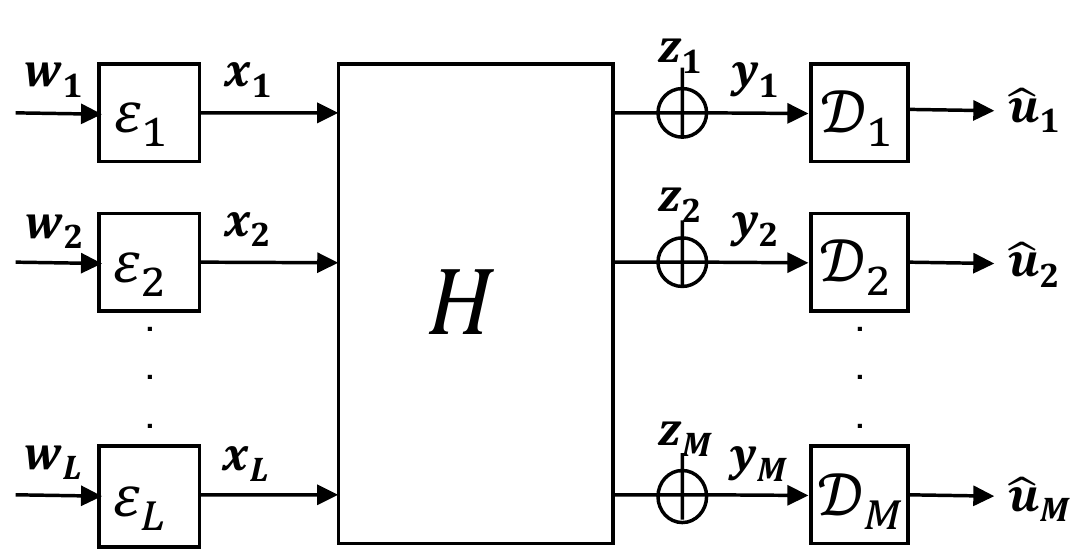}
\caption{Compute and Forward system model. $L$ transmitters communicate through a shared medium to $M$ relays.}
\label{fig-CF_System_model}
\end{figure}

\section{System model and Known results} \label{sec-System model}
Consider a multi-user multi-relay network, where $L$ transmitters are communicating to a single destination $D$ via $M$ relays. The model is illustrated in figure \ref{fig-CF_System_model}. All relays form a layer between the transmitters and the destination such that each transmitter can communicate with all the relays. Each transmitter draws a length-$k$ message with equal probability over a prime size finite field, $\mathbf{w}_l \in \mathbb{F}_p^k, \ l=1,2,...,L$, where $\mathbb{F}_p$ denotes the finite field with a set of $p$ elements.

This message is then forwarded to the transmitter's encoder, $\mathcal{E}_l:\mathbb{F}_p^k \rightarrow \mathbb{R}^n$, which maps length-$k$ messages over the finite field to length-$n$ real-valued codewords, $\mathbf{x}_l=\mathcal{E}_l(\mathbf{w_l})$. Each codeword is subject to a power constraint, $\|\mathbf{x}_l\|^2 \leq nP$. The message rate of each transmitter is defined as the length of the message measured in bits normalized by the number of channel uses, that is, $R=\frac{k}{n}\log{p}$, which is equal\footnote{Note that messages with different length can be allowed with zero padding to attain a length-$k$ message which will result in different rates for the transmitters.} for each transmitter. Each transmitter then broadcasts it's codeword to the channel.

Hence, each relay $m\in\{1...M\}$ observes a noisy linear combination of the transmitted signals through the channel,
\begin{equation}
\mathbf{y}_m=\sum_{i=1}^{L}h_{ml}\mathbf{x}_l+\mathbf{z}_m \ \ \ \ m=1,2,...,M,
\end{equation}

where $h_{ml} \sim \mathcal{N}(0,1)$ are the real channel coefficients and $\mathbf{z}$ is an i.i.d., Gaussian noise, $\mathbf{z} \sim \mathcal{N}(0,\mathbf{I}^{n\times n})$. Let $\mathbf{h}_m= [h_{m1},h_{m2},...,h_{mL}]^T$ denote the vector of channel coefficients at relay $m$. We assume that each relay knows its own channel vector. After receiving the noisy linear combination, each relay selects a scale coefficient $\alpha_m \in \mathbb{R}$, an integer coefficients vector $\b{a}_m=(a_{m1},a_{m2},...,a_{mL})^T \in \mathbb{Z}^L$, and attempts to decode the lattice point $\sum_{l=1}^La_{ml}x_l$ from $\alpha_m\b{y}_m$.

In CF each relay decodes a linear combination $\b{u}_m$ of the original messages, and forward it to the destination. With enough linear combinations, the destination is able to recover the desired (original) messages from all sources. 

The main results in CF are the following.
\begin{theorem}[{\cite[Theorem 1]{nazer2011compute}}]\label{the-Computation rate}
For real-valued AWGN networks with channel coefficient vectors $\b{h}_m \in \mathbb{R}^L$ and coefficients vector $\b{a}_m \in \mathbb{Z}^L$, the following computation rate region is achievable:
\begin{equation}
\cR(\b{h}_m,\b{a}_m)= \max \limits_{\alpha_m \in \R} \frac{1}{2} \log^+ \left( \frac{P}{\alpha_m^2+P\|\alpha_m \b{h}_m-\b{a}_m\|^2} \right),
\end{equation}
\end{theorem}
where $\log^+(x)\triangleq \max \{\log(x),0\}$.
\begin{theorem}[{\cite[Theorem 2]{nazer2011compute}}]\label{the-Computation rate with MMSE}
The computation rate given in Theorem \ref{the-Computation rate} is uniquely maximized by choosing $\alpha_m$ to be the MMSE coefficient
\begin{equation}
\alpha_{MMSE}=\frac{P\b{h}_m^T\b{a}_m}{1+P\|\b{h_m}\|^2},
\end{equation}
which results in a computation rate region of
\begin{equation}\label{equ-Computation rate with MMSE}
\cR(\b{h}_m,\b{a}_m)= \frac{1}{2} \log^+ \left( \|\b{a}_m\|^2- \frac{P(\b{h}_m^T\b{a}_m)^2}{1+P\|\b{h_m}\|^2} \right)^{-1}.
\end{equation}
\end{theorem}

Note that the above theorems are for real channels and the rate expressions for the complex channel are twice the above (\cite[Theorems 3 and 4]{nazer2011compute}).

Since the relay can decide which linear combination to decode (i.e., the coefficients vector $\textbf{a}$), an optimal choice will be one that maximizes the achievable rate. That is,
\begin{equation}\label{equ-Optimal a vector definition}
\textbf{a}_m^{opt}=\argmax_{\textbf{a}_m\in\mathbb{Z}^L \backslash \{\textbf{0}\}}\frac{1}{2}\log^+\left(\|\textbf{a}_m\|^2-\frac{P(\mathbf{h}_m^T\mathbf{a}_m)^2}{1+P\|\mathbf{h}_m\|^2} \right)^{-1}.
\end{equation}

\begin{remark}[The coefficients vector]
The coefficients vector $\b{a}$ plays a significant role in the CF scheme. It dictates which linear combination of the transmitted codewords the relay wishes to decode. That is, each non-zero element signifies the fact that the relay is interested in it's corresponding codeword. If, starting from a certain number of simultaneously transmitting users, the coefficients vector the relay chooses is always (or with high probability) a unit vector, this means that essentially we treat all other users as noise and loose the promised gain of CF. 
\end{remark}

The following Lemma bounds the search domain for the maximization problem in \eqref{equ-Optimal a vector definition}.
\begin{lemma}[{\cite[Lemma 1]{nazer2011compute}}]\label{lem-Search domain for the vector coefficients}
For a given channel vector $\b{h}$, the computation rate $\cR(\b{h}_m,\b{a}_m)$ in Theorem \ref{the-Computation rate with MMSE}  is zero if the coefficient vector $\b{a}$ satisfies
\begin{equation}\label{equ-Search domain for the vector coefficients}
\|\b{a}_m\|^2 \geq 1+P\|\b{h}_m\|^2.
\end{equation}
\end{lemma}
The problem of finding the optimal $\textbf{a}$ can be done by exhaustive search for small values of $L$. However, as $L$ grows, the problem becomes prohibitively complex quickly. In fact, it becomes a special case of the lattice reduction problem, which has been proved to be NP-complete. This can be seen if we write the maximization problem of \eqref{equ-Optimal a vector definition} as an equivalent minimization problem \cite{sahraei2014compute}:

\begin{equation}\label{equ-Optimal a vector in quadratic form}
\textbf{a}_m^{opt}=\argmin_{\textbf{a}_m\in\mathbb{Z}^L \backslash \{\textbf{0}\}} f(\textbf{a}_m)=\textbf{a}_m^T\textbf{G}_m\textbf{a}_m,
\end{equation}

where $\textbf{G}_m=(1+P\|\b{h}_m\|^2)\b{I}-P\b{h}_m\b{h}_m^T$. $\b{G}_m$ can be regarded as the Gram matrix of a certain lattice and $\b{a}_m$ will be the shortest basis vector and the one which minimize $f$. This problem is also known as the \emph{shortest lattice vector} problem (SLV), which has known approximation algorithms due to its hardness \cite{dadush2011enumerative, alekhnovich2005hardness}. The most notable of them is the LLL algorithm \cite{lenstra1982factoring, gama2008finding} which has an exponential approximation factor which grows with the size of the dimension. However, for special lattices, efficient algorithms exist \cite{conway2013sphere}. In \cite{sahraei2014compute}, a polynomial complexity algorithm was introduced for the special case of finding the best coefficient vector in CF.

%%%%%%%%%%%%%%%%%%%%%%%%%%%%%%%%%%%%%%%%%%%%%%%%%%%%%%%%%%%%%%%%%%%%%%%%%%%%%%%%%%%%%%

\section{Probability of a Unit Vector}\label{sec-Probability of a unit vector}
In this section, we examine the coefficient vector at a single relay, hence, we omit the index $m$ in the expressions. 

\subsection{The Matrix $\b{G}$}
Examining the matrix $\b{G}$, one can notice that as $L$, the number of transmitters, grows, the diagonal elements grow very fast relatively to the off-diagonal elements. Specifically, each diagonal element is a random variable, which is a $\chi^2_L$ r.v. minus a multiplication of two Gaussian r.vs., whereas the off-diagonal elements are only a multiplication of two Gaussian r.vs.. Of course, as $L$ grows, the former has much higher expectation value compared to later. Examples of $\b{G}$ are presented in Figure \ref{fig-Gmatrix}, for different dimensions. It is clear that even for moderate number of transmitters, the differences in values between the diagonal and off-diagonal elements are significant.

\begin{figure}[t]
    \centering
    \begin{subfigure}[b]{0.23\textwidth}
        \includegraphics[width=0.95\textwidth]{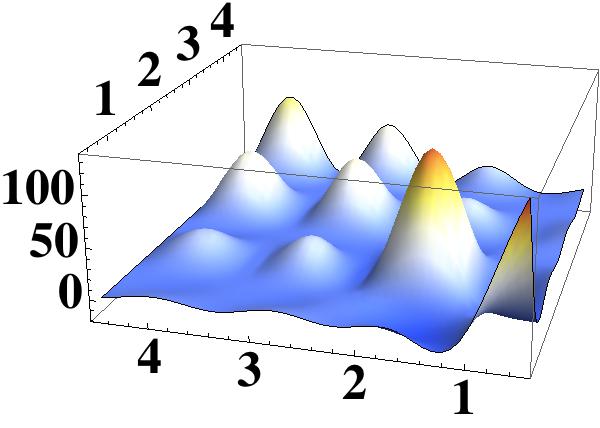}
        \caption{$L=4$}
    \end{subfigure}
    ~
    \centering
    \begin{subfigure}[b]{0.23\textwidth}
        \includegraphics[width=0.95\textwidth]{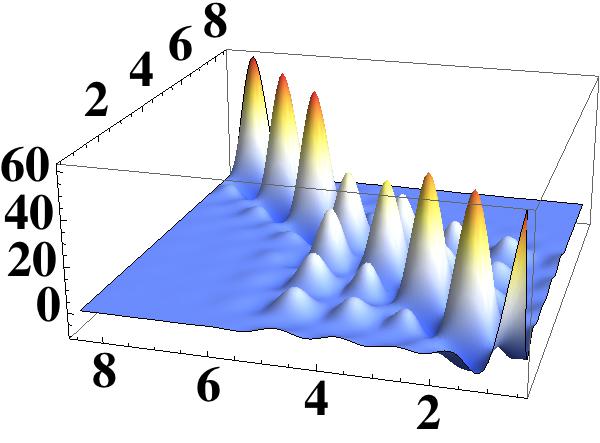}
        \caption{$L=8$}
    \end{subfigure}
    ~
    \begin{subfigure}[b]{0.23\textwidth}
        \includegraphics[width=0.95\textwidth]{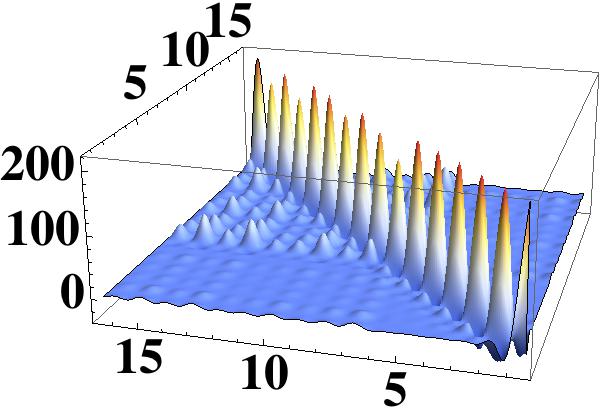}
        \caption{$L=16$}
    \end{subfigure}
    ~
    \begin{subfigure}[b]{0.23\textwidth}
        \includegraphics[width=0.95\textwidth]{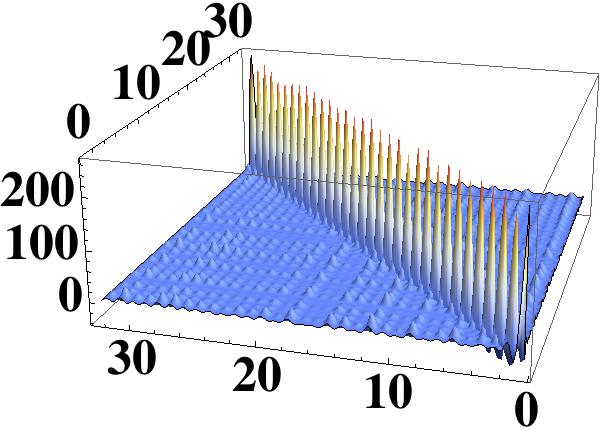}
        \caption{$L=32$}
    \end{subfigure}
    \caption{Example for the magnitude of the elements of $G$ for different dimensions (i.e., different values of $L$), for P=10. The graphs depict a single realization for each $L$, and were interpolated for ease of visualization.}
    \label{fig-Gmatrix}
\end{figure}

Consider now the quadric form \eqref{equ-Optimal a vector in quadratic form} we wish to minimize. Any choice of $\b{a}$ that is not a unit vector, will add more than one element from the diagonal of $\b{G}$ to it. When $L$ is large, the off-diagonal elements have little effect on the function value compared to the diagonal elements. Therefore, intuitively, one would prefer to have as little as possible elements from the diagonal although the off-diagonal elements can reduce the function value. This will happen if we choose $\b{a}$ to be a unit vector.  In the reminder of this section, we make this argument formal.

\subsection{Minimization of the Quadratic form $f$}
The minimization function $f(\b{a})=\textbf{a}^T\textbf{G}\textbf{a}$ can be written as 

\begin{equation*}
\begin{aligned}
\b{a^TGa}&=\sum_{i=1}^{L}(1-P(\|\b{h}\|^2-h_i^2))a_i^2   -  2\sum_{i=1}^{L}\sum_{j=1}^{i-1}Ph_ih_ja_ia_j \\
	       &=\|\b{a}\|^2 + P\sum_{i=1}^{L}\sum_{j=1}^{i-1} (h_ia_j-h_ja_i)^2\\
	       &=\|\b{a}\|^2 + P(\|\b{a}\|^2\|\b{h}\|^2-(\b{a}^T\b{h})^2).
\end{aligned}
\end{equation*}

Note that the right term consists of all possible pairs $(i,j)$ such that $i \neq j$, a total of $\frac{L(L-1)}{2}$ elements.  

We wish to understand when will a relay prefer a unit vector over any other non-trivial vector $\b{a}$. Specifically, since $\b{a}$ is a function of the random channel $\b{h}$, we will compute the probability of having a unit vector as the minimizer of $f$ for a given $\b{a}$. Or, alternatively, the probability that a certain non-trivial $\b{a}$ will minimize $f$ compared to a unit vector. We thus wish to find the probability

\small
\begin{equation}\label{equ-Probability of choosing a unit vector}
\begin{aligned}
	&Pr( f(\b{a}) \leq f(e_i))\\
&=Pr( \|\b{a}\|^2 + P(\|\b{a}\|^2\|\b{h}\|^2-(\b{a}^T\b{h})^2) \leq 1+P(\|\b{h}\|^2-h_i^2)),
\end{aligned}
\end{equation}
\normalsize

where $e_i$ is a unit vector of size $L$ with $1$ at the $i$-th entry and zero elsewhere, and $\b{a}$ is \emph{any} integer valued vector that is \emph{not} a unit vector. Note that \eqref{equ-Probability of choosing a unit vector} refers to any integer vector $\b{a}$, including the vectors in the search domain such that $\|\textbf{a}\|\leq\sqrt{1+P\|\textbf{h}\|^2}$ (\cite[Lemma 1]{nazer2011compute}). 

Note also that the right and left hand sides of the inequality in equation \eqref{equ-Probability of choosing a unit vector} are dependent, hence direct computation of this probability is not trivial. Still, this probability can be evaluated exactly noting that the angle between $\b{a}$ and $\b{h}$ is what mainly affects it. The details are in the theorem below.

\subsection{The Optimality of $e_i$ VS. a Certain Vector $\b{a}$}

\begin{theorem}\label{the-Probability for having a unit vector as the maximaizer}
Under the CF scheme, the probability that a non-trivial vector $\b{a}$ will be the coefficient vector $\b{a}^{opt}$ which maximize the achievable rate $\mathcal{R}(\b{h},\b{a}^{opt})$, i.e., minimize $f(\b{a}^{opt})$,  comparing with a unit vector $e_i$, is upper bounded by
\begin{equation}\label{equ-Probability of unit vector as minimizer of f}
Pr( f(\b{a}) \leq f(e_i)) \leq 1- I_{\Phi(\b{a})}\left(\frac{1}{2},\frac{L-1}{2}\right)
\end{equation}
where $I_x(a,b)$ is the CDF of the Beta distribution with parameters $a$ and $b$, and $\Phi(\b{a})=1-\frac{1}{\|a\|^2}$. Note that $\frac{1}{2} \leq \Phi(\b{a}) \leq 1$ for any $\b{a}$ which is not a unit vector.
\end{theorem}

In the context of this work, the main consequence of Theorem \ref{the-Probability for having a unit vector as the maximaizer}, is the following. 

\begin{corollary}\label{cor-Probability for having a unit vector as the maximaizer goes to one}
 As the number of simultaneously transmitting users grows, the probability that a non-trivial $\b{a}$ will be the maximizer for the achievable rate goes to zero. Specifically, 
 \begin{equation}\label{equ-Probability of choosing a unit vector goes to one}
Pr( f(\b{a}) \leq f(e_i)) \leq e^{-LE_1(L)},
\end{equation}
where $\b{a}$ is any integer vector that is \emph{\textbf{not}} a unit vector and  $E_1(L)=(1-\frac{3}{L})\log{\|\b{a}\|}$.
\end{corollary}
The proofs will be given after the following discussion.

\subsubsection{Discussion and simulation results}
 Corollary \ref{cor-Probability for having a unit vector as the maximaizer goes to one} clarifies that for every $P$, as the number of users grows, the probability of having a non-trivial vector $\b{a}$ as the maximizer of the achievable rate is going to 0. Note that the assumption of $L>3$, which arises naturally form this paper's regime, along with the fact that $\|\b{a}\|\geq2$ grantees that $E_1(L)$ is positive.  Figure \ref{fig-Probability_for_having_unit_vector} depicts the probability in \eqref{equ-Probability of unit vector as minimizer of f}, it's upper bound given in equation \eqref{equ-Probability of choosing a unit vector goes to one} and simulation results. From the analytic results as well as the simulations on the rate of decay, one can deduce that even for relatively small values of simultaneously transmitting users ($L>20$), the relay will prefer to choose a unit vector. Also, one can observe from the results and from the analytic bound that as the norm of $\b{a}$ grows, the rate of decay increases. This faster decay reflects the increased penalty of approximating a real vector using an integer valued vector.

\subsubsection{Proofs} 
The proof of Theorem \ref{the-Probability for having a unit vector as the maximaizer} is based on the lemma below.
\begin{lemma}\label{lem-distribution of the squared cosine of the phase}
The distribution of $\frac{(\b{a}^T\b{h})^2}{\|\b{a}\|^2\|\b{h}\|^2}$, which is the squared cosine of the angle between an integer vector $\b{a}$ and a standard normal vector $\b{h}$, both of dimension L, is $Beta(\frac{1}{2},\frac{L-1}{2})$.
\end{lemma}
\begin{proof}
Let $\b{Q}$ be an orthogonal rotation matrix such that $\b{Q}\b{a}=\b{a'}$ where $\b{a'}$, is co-linear to the basis vector $e_1$. That is, $\b{a'}=(\|\b{a}\|,0,...,0)$. Define $\b{h'}=\b{Q}\b{h}$. Note that $\b{h'}$ is a standard normal vector since $E[\b{h'}]=E[\b{Q}\b{h}]=0$, and $\b{QIQ}^T=\b{QQ}^T=\b{I}$. We have
\begin{equation}
\begin{aligned}
	&\frac{(\b{a}^T\b{h})^2}{\|\b{a}\|^2\|\b{h}\|^2}=\frac{(\b{a}^T\b{h})^2}{(\b{a}^T\b{a})(\b{h}^T\b{h})}=\frac{(\b{a}^T\b{Q}^T\b{Qh})^2}{(\b{a}^T\b{Q}^T\b{Qa})(\b{h}^T\b{Q}^T\b{Qh})}\\
	&=\frac{((\b{Qa})^T\b{Qh})^2}{((\b{Qa})^T\b{Qa})((\b{Qh})^T\b{Qh})}=\frac{((\b{Qa})^T\b{Qh})^2}{\|\b{Qa}\|^2\|\b{Qh}\|^2}\\
	&=\frac{\|\b{a'}\|^2(\b{e_1}^T\b{h'})^2}{\|\b{a'}\|^2\|\b{e_1}\|^2\|\b{h'}\|^2}=\frac{(\b{e_1}^T\b{h'})^2}{\|\b{e_1}\|^2\|\b{h'}\|^2}=\frac{{h'}_1^2}{\|\b{h'}\|^2}.
\end{aligned}
\end{equation}

Considering the above we have the equality of  $\cos^2{\theta}=\frac{{h'}_1^2}{{h'}_1^2+{h'}_2^2+...+{h'}_L^2}$. This expression can be represented as $\frac{W}{W+V}$, where $W={h'}_1^2$ is a $\chi^2_1$ r.v. and $V=\sum_{i=2}^L{h'}_i^2$ is a $\chi^2_{L-1}$ r.v. independent in $W$. This ratio has a $Beta(a,b)$ distribution, with $a=\frac{1}{2}$ and $b=\frac{L-1}{2}$. Note that $a$ and $b$ correspond to the degrees of freedom of $W$ and $V$.

\end{proof}

\begin{figure}[t]
\centering
    \includegraphics[width=0.5\textwidth]{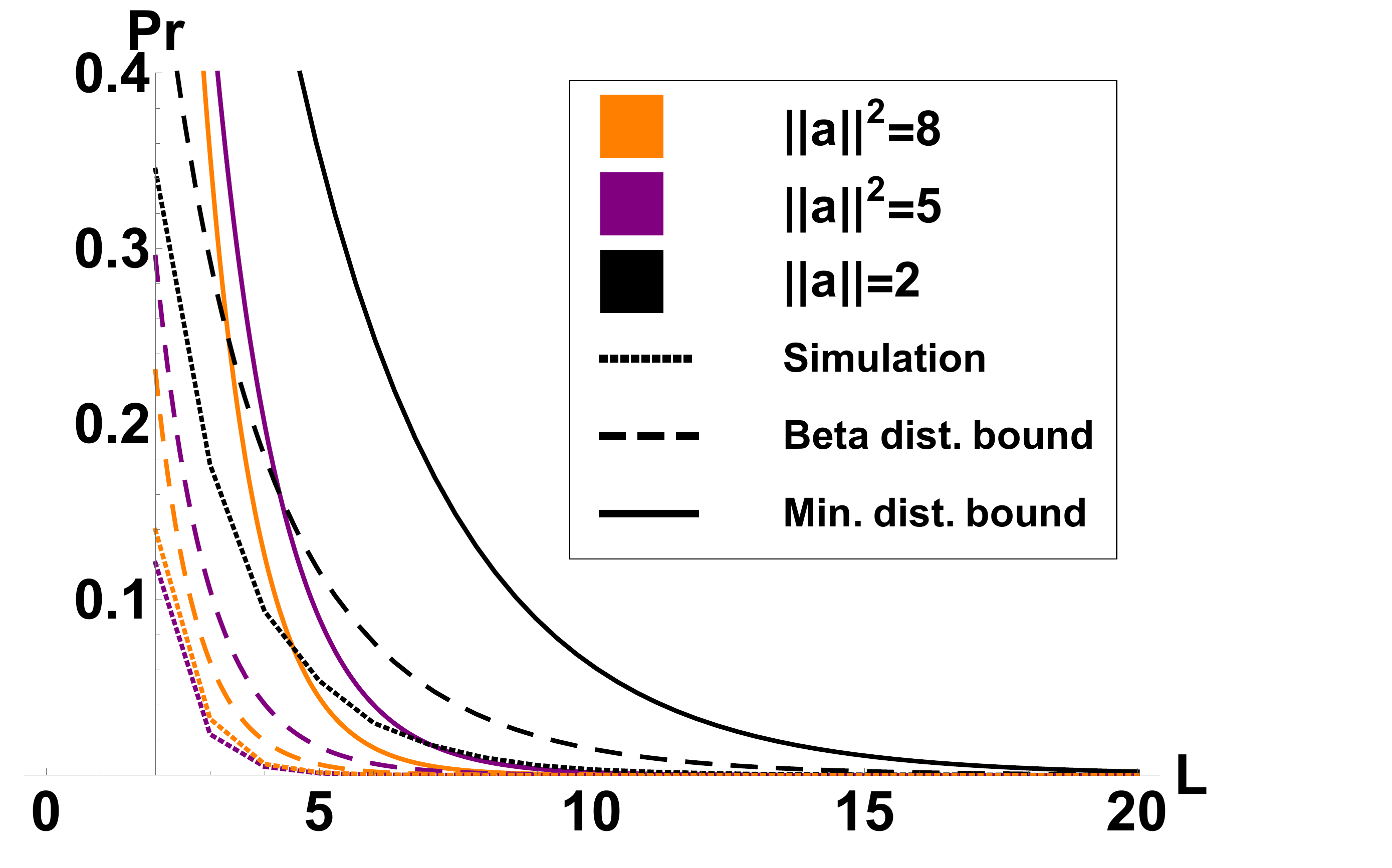}
\caption{The upper bounds given in \eqref{equ-Probability of choosing a unit vector goes to one}(solid lines), \eqref{equ-Probability of unit vector as minimizer of f} (dashed lines)  and simulation results (dotted lines) for not having a unit vector as the minimizer of $f$ compared to various values of $\|\b{a}\|^2$ as a function of simultaneously transmitting users.}
\label{fig-Probability_for_having_unit_vector}
\end{figure}

\begin{myproof} [Proof of Theorem \ref{the-Probability for having a unit vector as the maximaizer}]
According to equation \eqref{equ-Probability of choosing a unit vector} we have,
\small
\begin{equation}
\begin{aligned}
	&Pr( f(\b{a}) \leq f(e_i))\\
	&=Pr\left(\|\b{a}\|^2 +P\left(\|\b{a}\|^2\|\b{h}\|^2-(\b{a}^T\b{h})^2\right)  \leq 
	 1+P(\|\b{h}\|^2-h_i^2) \right)\\
	&=Pr\left(\frac{1-\|\b{a}\|^2}{P} + \|\b{h}\|^2 - h_i^2 - \|\b{a}\|^2\|\b{h}\|^2 + (\b{a}^T\b{h})^2 \geq 0 \right)\\
	&\overset{(a)}{\leq} Pr\left(\|\b{h}\|^2 - \|\b{a}\|^2\|\b{h}\|^2 + (\b{a}^T\b{h})^2 \geq 0 \right)\\
	&=Pr\left(\frac{1}{\|\b{a}\|^2} - 1 + \frac{(\b{a}^T\b{h})^2}{\|\b{a}\|^2\|\b{h}\|^2} \geq 0 \right)\\
	&=Pr\left(\frac{(\b{a}^T\b{h})^2}{\|\b{a}\|^2\|\b{h}\|^2} \geq 1- \frac{1}{\|\b{a}\|^2} \right)\\
	&\overset{(b)}{=}1- I_{\Phi(\b{a})}\left(\frac{1}{2},\frac{L-1}{2}\right),
\end{aligned}
\end{equation}
\normalsize
where $(a)$ follows since we removed negative terms ($\|\b{a}\|^2>1$) and $(b)$ follows from Lemma \ref{lem-distribution of the squared cosine of the phase} with $\Phi(\b{a})=1- \frac{1}{\|\b{a}\|^2}$.
\end{myproof}

The bound on the probability given in Theorem \ref{the-Probability for having a unit vector as the maximaizer} consists of a complicated analytic function $I_{\Phi(\b{a})}(\cdot)$. Hence, corollary \ref{cor-Probability for having a unit vector as the maximaizer goes to one} includes a simplified bound which avoids the use of $I_{\Phi(\b{a})}(\cdot)$, yet keeps the nature of the result in Theorem \ref{the-Probability for having a unit vector as the maximaizer}. The proof of Corollary \ref{cor-Probability for having a unit vector as the maximaizer goes to one} is based on the following lemma.

\begin{lemma}\label{lem-lower bound for the distribution of the squared cosine of the phase}
The CDF of $\frac{(\b{a}^T\b{h})^2}{\|\b{a}\|^2\|\b{h}\|^2}$ %, which is the squared cosine of the phase between an integer vector $\b{a}$ and a standard normal vector $\b{h}$, both of dimension L, 
can be lower bounded by the CDF of the minimum of $\left(\floor*{\frac{L}{2}}-1\right)$ i.i.d. uniform random variables in $[0,1]$.
\end{lemma}
\begin{proof}
We start by assuming that $L$ is even where the case of odd $L$ will be dealt with later. From Lemma \ref{lem-distribution of the squared cosine of the phase}, the r.v. $\frac{(\b{a}^T\b{h})^2}{\|\b{a}\|^2\|\b{h}\|^2}$ has the same distribution as $\frac{h_1^2}{\|\b{h}\|^2}$, that is, for any $0 \leq \alpha \leq 1$,
\small
\begin{align*}
	&Pr\left(\frac{(\b{a}^T\b{h})^2}{\|\b{a}\|^2\|\b{h}\|^2} \leq \alpha \right)\\
	&=Pr\left(\frac{h_1^2}{\|\b{h}\|^2} \leq \alpha \right)\\
	&\overset{(a)}{\geq} Pr\left(\frac{h_1^2 + h_2^2}{\|\b{h}\|^2} \leq \alpha \right)\\
	&= Pr\left(\frac{h_1^2 + h_2^2}{(h_1^2 + h_2^2)+...+(h_{L-1}^2 + h_L^2)} \leq \alpha \right)\\
	&\overset{(b)}{=} 1-(1-\alpha)^{\frac{L}{2}-1},
\end{align*}
\normalsize
$(a)$ is true since a larger r.v. will yield lower probability. $(b)$ is due to the observation that $\frac{h_1^2 + h_2^2}{\|\b{h}\|^2}$ can be represented as $\frac{W}{W+V}$, where $W=h_1^2 + h_2^2$ and $V=\sum_{i=3}^{L} h_i^2$ are independent exponential r.vs. Note that $V$ is essentially a sum of $\frac{L}{2}-1$ independent pairs. This ratio is distributed as the minimum of $\left(\frac{L}{2}-1\right)$ i.i.d. uniform $[0,1]$ random variables \cite[Lemma 3.2]{jagannathan2006efficient}. This is since the ratio can be interpreted as the proportion of the waiting time for the first arrival to the $\frac{L}{2}$ arrival of a Poisson process. 

%W has parameter $\frac{1}{2}$ and V has parameter $\frac{}{}$... which consists of $L/2$ independent pairs

In case $L$ is an odd number, we can increase the term in the proof by replacing it with $\frac{h_1^2 + h_2^2}{\|\b{h}\|^2 - h_L^2}$, resulting in a distribution which is similar to the minimum of $(\frac{L-1}{2}-1)$ i.i.d. uniform random variables in the same manner.
\end{proof} 

\begin{myproof} [Proof of Corollary \ref{cor-Probability for having a unit vector as the maximaizer goes to one}]
\small
\begin{equation}\label{equ-upper bound on the probability for not having a unit vector}
\begin{aligned}
	& Pr(f(\b{a}) \leq f(e_i))\\
	&\overset{(a)}{\leq} Pr\left(\frac{(\b{a}^T\b{h})^2}{\|\b{a}\|^2\|\b{h}\|^2} \geq 1- \frac{1}{\|\b{a}\|^2} \right)\\
	&\overset{(b)}{\leq} \left(1-\left(1- \frac{1}{\|\b{a}\|^2} \right) \right)^{\floor*{\frac{L}{2}}-1}\\
	&\leq   \left(\frac{1}{\|\b{a}\|^2} \right)^{\frac{L-1}{2}-1}\\
	&=   e^{-LE_1(L)},
\end{aligned}
\end{equation}
\normalsize

where $(a)$ and $(b)$ follows from Lemmas \ref{lem-distribution of the squared cosine of the phase} and \ref{lem-lower bound for the distribution of the squared cosine of the phase} respectively, $E_1(L)=(1-\frac{3}{L})\log{\|\b{a}\|}$ and $\|\b{a}\|^2>1$. %In case $L$ is odd E_1(L) we can conclude that the probability goes to zero exponentially as L grows. 
\end{myproof}

The following lemma shows a simple property of the optimal coefficients vector which shows that if the relay is interested in only one transmitter, i.e. a unit vector as the optimal coefficients vector, it will be the transmitter with the strongest channel.

\begin{lemma}\label{lem-Position of maximal value in a vector}
For any channel vector $\b{h}$ of size $L$ with $\mathrm{m}=\argmax_i\{|h_i|\}$, the optimal coefficients vector $\b{a}^{opt}$ which maximize the rate $\mathcal{R}(\textbf{h},\textbf{a})$ has to satisfy $\mathrm{m}=\argmax_i\{|a_i|\}$ as well.
\end{lemma}
\begin{proof}
Suppose that there exist $\b{h}$ for which (w.l.o.g.) $|h_1|>|h_i|$ for all $i\neq1$ and that the optimal coefficients vector $\b{a}^{opt}$ satisfies (w.l.o.g.) $|a_2|>|a_i|$ for all $i\neq2$. Considering the rate expression $\mathcal{R}(\textbf{h},\textbf{a})$ we will show that by rearranging $\b{a}^{opt}$ a higher rate can be attain. Let $\widehat{\b{a}}$ be a vector which is identical to $\b{a}^{opt}$ except the two first entries which are switched, i.e. $\widehat{a_1}=a_2^{opt}$ and $\widehat{a_2}=a_1^{opt}$. The values of the vectors are the same thus we have $\|\widehat{\b{a}}\|=\|\b{a}^{opt}\|$ and the only term affecting the rate is the scalar multiplication between $\b{a}^{opt}$ and $\b{h}$. We first note that the signs of $h_i$ and its corresponding optimal coefficient $a_i^{opt}$ has to be equal or different for all $i$. I.e. the the case which there exist $i$ such that $sign(h_i)=sign(a_i^{opt})$ and $j$ such that $sign(h_j) \neq sign(a_j^{opt})$ could not be possible. This is due to the fact that the optimal coefficients vector has to maximize the scalar multiplication. Therefore, considering this property, $h_1\widehat{a_1}+h_2\widehat{a_2}>h_1a_1^{opt}+h_2a_2^{opt}$. This means that the rate can be improved by choosing $\widehat{\b{a}}$ contradicting $\b{a}^{opt}$ optimality. Specifically, we'll get that as long as the maximal value in any $\b{a}$ is not in the same place as the maximal value in $\b{h}$ we can always improve the rate. 
\end{proof}

\subsection{The Optimality of $e_i$ VS. All Possible Vectors $\b{a}$}
Corollary \ref{cor-Probability for having a unit vector as the maximaizer goes to one} refers to the probability that a unit vector will minimize $f$ for a fixed $\b{a}$. Next we wish to explore this probability for any possible $\b{a}$. For the purpose of clarity, \eqref{equ-upper bound on the probability for not having a unit vector} gives an upper bound on the probability that a unit vector will not minimize $f$ compered to a certain possible integer coefficients vectors with a certain $\|\b{a}\|^2$. Where the probability of having an optimal vector which is \emph{not} a unit vector will be the union of all probabilities for each $\b{a}$ vector which satisfies $\|\b{a}\|^2< 1+P\|\b{h}\|^2$. Let us define $P(e_i)$ as the probability that a relay picked a unit vector as the coefficient vector, and $P(\overline{e_i})$ as the probability which any other vector was chosen.

%Next we would like to explore the probability of any other vector than $e_\mathrm{m}$ to be chosen as the optimal vector as a function of $L$. Let us define $P(e_\mathrm{m})$ as the probability that a relay picked a unit vector as the coefficient vector, and $P(\overline{e_\mathrm{m}})$ as the probability which any other vector was chosen. Thus, as mentioned, $P(\overline{e_\mathrm{m}})$ will be the union of all probabilities for each  $\b{a}$ vector which satisfies $\|\b{a}\|^2< 1+P\|\b{h}\|^2$ and can be considered as an option for the optimization problem.  

In \cite{sahraei2014compute}, a polynomial time algorithm for finding the optimal coefficients vector $\b{a}$ was given. The complexity result derives from the fact that the \textit{cardinality of the set of all $\b{a}$ vectors} (denoted as $\Phi$)\footnote{In \cite{sahraei2014compute}, $\Phi$ is a set of points which the average of any consecutive points is mapped to a different coefficients vector.} which are considered is upper bounded by $2L(\ceil{\sqrt{1+P\|\b{h}\|^2}}+1)$. That is, any vector which does not exist in this set has zero probability to be the one which maximize the rate. We shell note this set here as $\mathcal{A}$. Thus, we wish to compute

\begin{equation}
P(\overline{e_i})= \bigcup_{\mathcal{A}} P \left(  f(\b{a}) \leq f(e_i) \right),
\end{equation}
where $\mathcal{A}= \{ \b{a} \in \mathbb{Z}^L : \b{a} \in \Phi, \ \b{a}\neq e_i \ \forall i \}$. Note that the cardinality of $\mathcal{A}$ grows with the dimension of $\b{h}$, i.e., with $L$ and can be easily upper bounded as follows,

\begin{multline}\label{equ-Cardinality of the search domain}
|\mathcal{A}| \leq 2L(\ceil{\sqrt{1+P\|\b{h}\|^2}}+1) \\
\leq 2L(\ceil{1+P\|\b{h}\|^2}+1) \leq  2L(P\|\b{h}\|^2+3).
\end{multline}

\begin{theorem}\label{the-Probability for having a unit vector as the maximaizer over all other vectors}
Under the CF scheme, the probability which any other coefficients vector $\b{a}$ will be chosen to maximize the achievable rate $\mathcal{R}(\b{h},\b{a}^{opt})$ compared with a unit vector $e_i$, as the number of simultaneously transmitting users grows, is zero. That is, 
\begin{equation}
\lim_{L \rightarrow \infty} P(\overline{e_i})= 0,
\end{equation}
\end{theorem}

\begin{proof}
We have, 
\begin{align*}
&\lim_{L \rightarrow \infty}P(\overline{e_i})= \lim_{L \rightarrow \infty} \bigcup_{\mathcal{A}} P\left(f(\b{a}) \leq f(e_i)\right) \\
&\leq \lim_{L \rightarrow \infty} \sum_{\mathcal{A}}  P\left(f(\b{a}) \leq f(e_i)\right)\\
&\leq \lim_{L \rightarrow \infty} \sum_{\mathcal{A}} \left(\frac{1}{\|\b{a}\|^2} \right)^{\frac{L-1}{2}-1}\\
&\overset{(a)}{\leq} \lim_{L \rightarrow \infty} |\mathcal{A}| \left(\frac{1}{2} \right)^{\frac{L-1}{2}-1}\\
&\overset{(b)}{\leq} \lim_{L \rightarrow \infty} 2L(P\|\b{h}\|^2+3) \left(\frac{1}{2} \right)^{\frac{L-1}{2}-1}\\
&= \lim_{L \rightarrow \infty} 2L(P\sum_{i=1}^Lh_i^2+3) \left(\frac{1}{2} \right)^{\frac{L-1}{2}-1}\\
&\overset{(c)}{=} \lim_{L \rightarrow \infty} 2L^2P\left(\frac{1}{2} \right)^{\frac{L-1}{2}-1}  \frac{1}{L}\sum_{i=1}^Lh_i^2\\
&\overset{(d)}{=} \lim_{L \rightarrow \infty} 2L^2P\left(\frac{1}{2} \right)^{\frac{L-1}{2}-1} \\
&=4P \lim_{L \rightarrow \infty}  L^2 2^{-\frac{L-1}{2}} \\
&=4P\lim_{L \rightarrow \infty}  L^2 e^{-LE_2(L)} = 0,
\end{align*}
where $(a)$ is true since the term inside the sum is maximized with $\|\b{a}\|^2=2$. $(b)$ is due to \eqref{equ-Cardinality of the search domain}, in $(c)$ we multiplied and divide with $L$ and eliminate the limit term which is multiplied by 3 since it goes to zero. $(d)$ follows from the strong law of large numbers were the normalized sum converge with probability one to the expected value of $\chi^2_1$ r.v. which is one. And lastly we define $E_2(L)=\frac{1}{2}(1-\frac{1}{L})\log{2}$.
\end{proof}
This result implies that the probability of having any non unit vector as the rate maximizer is decreasing exponentially to zero as the number of users grows.

\section{Compute and Forward Sum-Rate}\label{sec-Compute and Forward Sum-Rate}
In order that relay $m$ will be able to decode a linear combination with coefficients vector $\b{a}_m$, all messages' rates which are involved in the linear combination must be within the computation rate region \cite{nazer2011compute}. I.e., all the messages for which the corresponding entry in the coefficient vector is non zero. That is,

\begin{equation}
R_l<\min_{a_{ml} \neq 0} \cR(\b{h}_m,\b{a}_m)
\end{equation}

Hence, the sum rate of the system is defined as the sum of messages' rates, i.e.,
\begin{equation}\label{equ-Sum rate expression}
\sum_{l=1}^{L} \min_{m:a_{ml} \neq 0} \cR(\b{h}_m,\b{a}_m).
\end{equation}

Following the results from previous subsections, we would like to show that as the number of users grows, the system's sum-rate decreases to zero as well. That is, without scheduling users, not only each individual rate is negligible, this is true for the sum-rate as well. This will strengthen the necessity to schedule users in CF.

\begin{figure}[t]
\centering
    \includegraphics[width=0.4\textwidth]{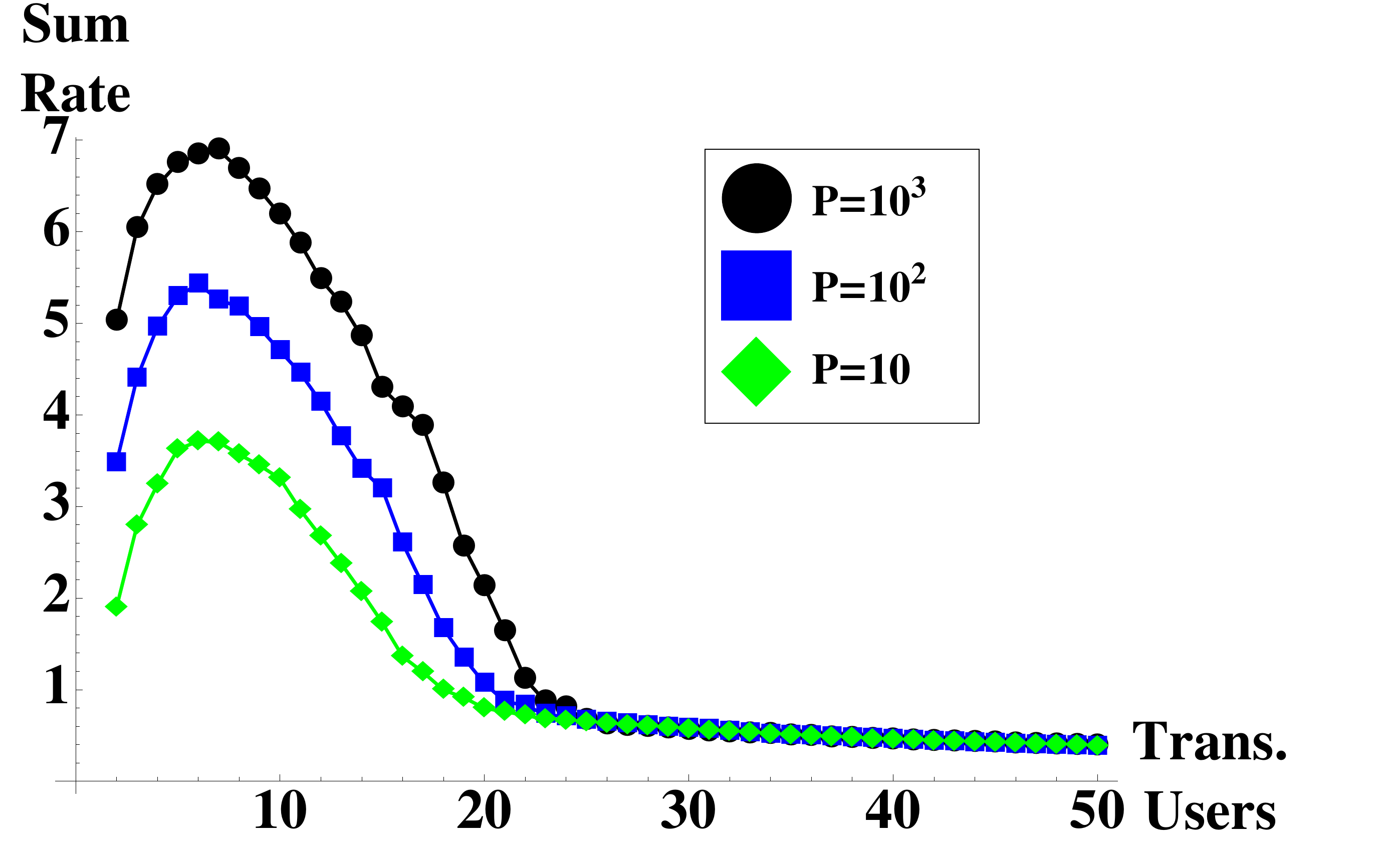}
\caption{The sum rate as give in \eqref{equ-Sum rate expression} for the case of 4 relays as a function of the number of simultaneously transmitting users, for different values of P.}
\label{fig-Rate_AsFunc_users}
\end{figure}

\begin{theorem}\label{the-Sum rate is going to zero}
As $L$ grows, the sum rate of CF is tends to zero, that is,
\begin{equation}
\lim_{L \rightarrow \infty} \sum_{l=1}^{L} \min_{m:a_{ml} \neq 0} \cR(\b{h}_m,\b{a}_m) = 0.
\end{equation}
\end{theorem}
\begin{proof}
The proof outline is as follows. The sum rate expression is divided into two parts, which describe two scenarios. The first is for the case where a relay chooses a unit vector as the coefficients vector and the second is for the case where any other vector is chosen. The probabilities for that are $P(e_i)$ and $P(\overline{e_i})$, respectively. Then, we show that each part goes to zero by upper bounding the corresponding expressions.
The complete proof is given in appendix \ref{AppendixA}.
\end{proof}

Simulations for the sum-rate for different values of P can be found in Figure \ref{fig-Rate_AsFunc_users}. It is obvious that for large $L$, the sum-rate decreases, hence, for a fixed number of relays there is no use in scheduling a large number of users, as CF degenerates to choosing unit vectors and treating other users as noise. However, the simulations suggests a peak at a small number of transmitters. We explore this in the next section.

%%%%%%%%%%%%%%%%%%%%%%%%%%%%%%%%%%%%%%%%%%%%%%%%%%%%%%%%%%%%%%%%%%%%%%%%%%%%%%%%%%%%%%

\section{Scheduling in Compute and Forward}\label{sec-Scheduling in CF}

Theorem \ref{the-Probability for having a unit vector as the maximaizer over all other vectors} and \ref{the-Sum rate is going to zero} suggest that a restriction on the number of simultaneously transmitting users should be made. That is, in order to apply the CF scheme for systems with a large number of sources, scheduling a smaller number of users should take place. 

\begin{figure}[t]
    \centering
        \includegraphics[width=0.5\textwidth]{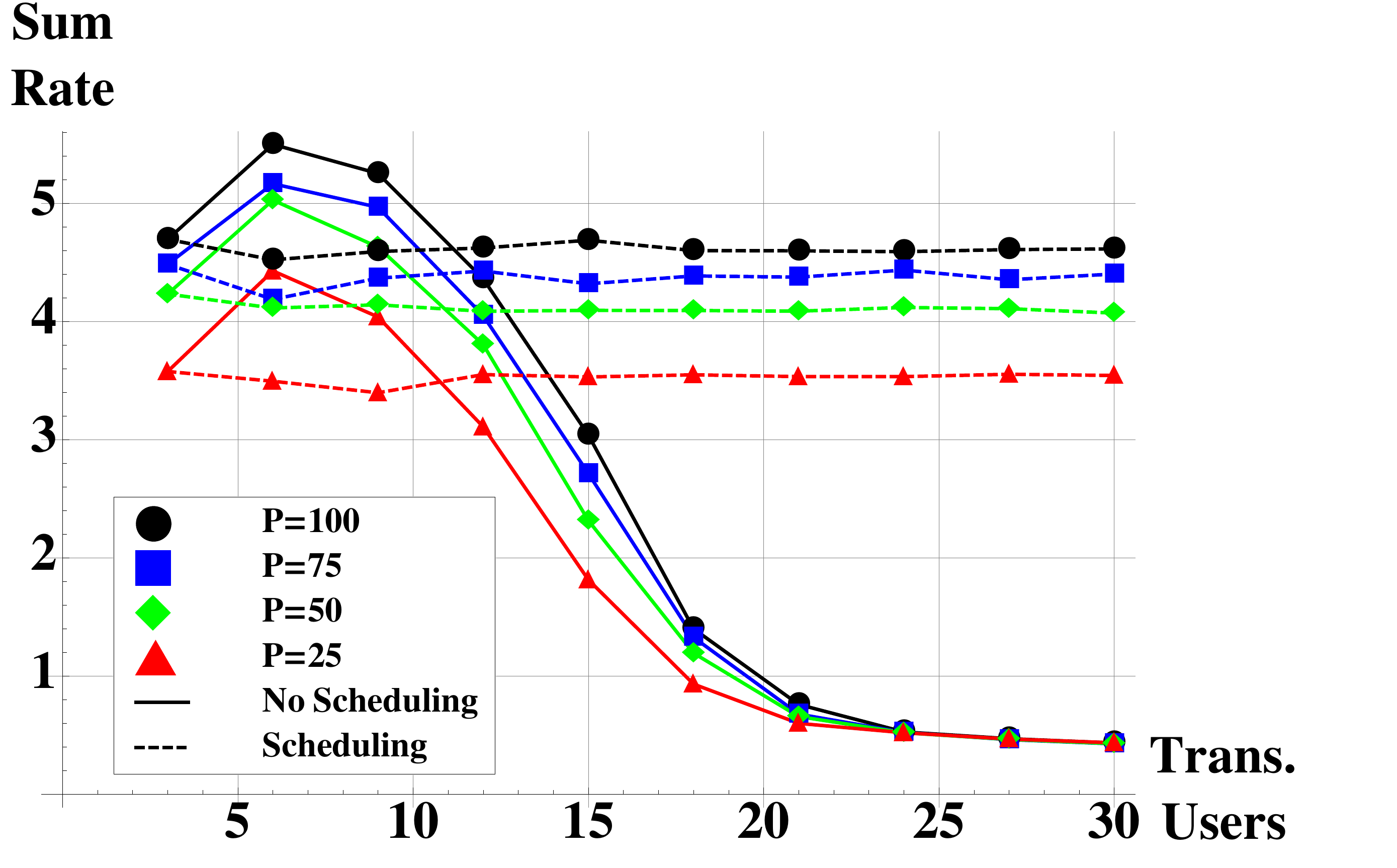}
    \caption{Simulation results for the average sum-rate per transmission. Here, the number of relays is 3 and scheduling was performed in a Round Robin manner, where in every phase 3 sources were scheduled among the transmitting users.}
       \label{fig-SumRate sch. and no sch. per slot}
\end{figure}

The most simple scheduling scheme is to schedule users in a Round Robin (RR) manner, where in each transmission only $k$ users may transmit simultaneously. The value of $k$ can be optimized, yet as a thump rule, one can schedule $M$ users (similar to the number of relays) is each transmission to obtain a sum-rate which is not going to zero. Figure \ref{fig-SumRate sch. and no sch. per slot} depicts such a scenario. In fact, even higher sum rates can be obtained if the number of scheduled users is \emph{higher than the number of relays}, i.e., the number for which the maximal sum-rate in Figure \ref{fig-Rate_AsFunc_users} is achieved. Still, it is clearly seen that it is not zero for $M$ scheduled sources, compared to the zero sum-rate when all $L$ users transmit and the relay use CF.

In fact, one can use existing results for the CF sum-rate for the case of equal number of sources and relays, and describe the sum-rate in each transmission under such a schedule. According to \cite{niesen2012degrees}, the sum rate for $M$ sources and $M$ relays is upper bounded by,

\begin{equation}\label{equ-Sum rate expression upper bounded}
\sum_{l=1}^{M} \min_{m:a_{ml} \neq 0} \cR(\b{h}_m,\b{a}_m) \leq \frac{1}{1+1/M}\log{P}+\log{\log{P}},
\end{equation}

for $M \geq 2$ and $P \geq 3$. A very coarse lower bound can be attained if the relays are forced to choose their coefficients vectors such that each relay $i$ chooses $a_i=e_i$. That is, an interference channel where each relay $i$ considers the interferences from all other sources $j, \ j \neq i$ as noise. Even with this one has

\small
\begin{equation}\label{equ-Sum rate expression lower bounded}
\begin{aligned}
&\sum_{l=1}^{M} \min_{m:a_{ml} \neq 0} \cR(\b{h}_m,\b{a}_m) \\
&\geq \sum_{l=1}^{M} \min_{m:a_{ml} \neq 0} \frac{1}{2} \log^+\left( \frac{1+P\|\b{h}_m\|^2}{1+P(\|\b{h}_m\|^2-h_l^2)} \right)\\
&\geq \sum_{l=1}^{M} \frac{1}{2}\log \left( 1+\frac{Ph_l^2}{1+P\sum \limits_{j\neq l}h_j^2} \right),
\end{aligned}
\end{equation}
\normalsize

which is not zero. Simulation results for the bounds and the optimal CF coefficient vectors are presented in Figure \ref{fig-SumRate Upper Lower As Func power}. From the aforesaid one can conclude that scheduling $M$ users for transmission is worthwhile with respect to the alternative of permitting all users transmit simultaneously. Of course, the scheduling policy has great impact on the performance, which can be increased if, for example, one schedules groups whose channel vectors are more suitable for CF.

\begin{figure}[t]
    \centering
        \includegraphics[width=0.5\textwidth]{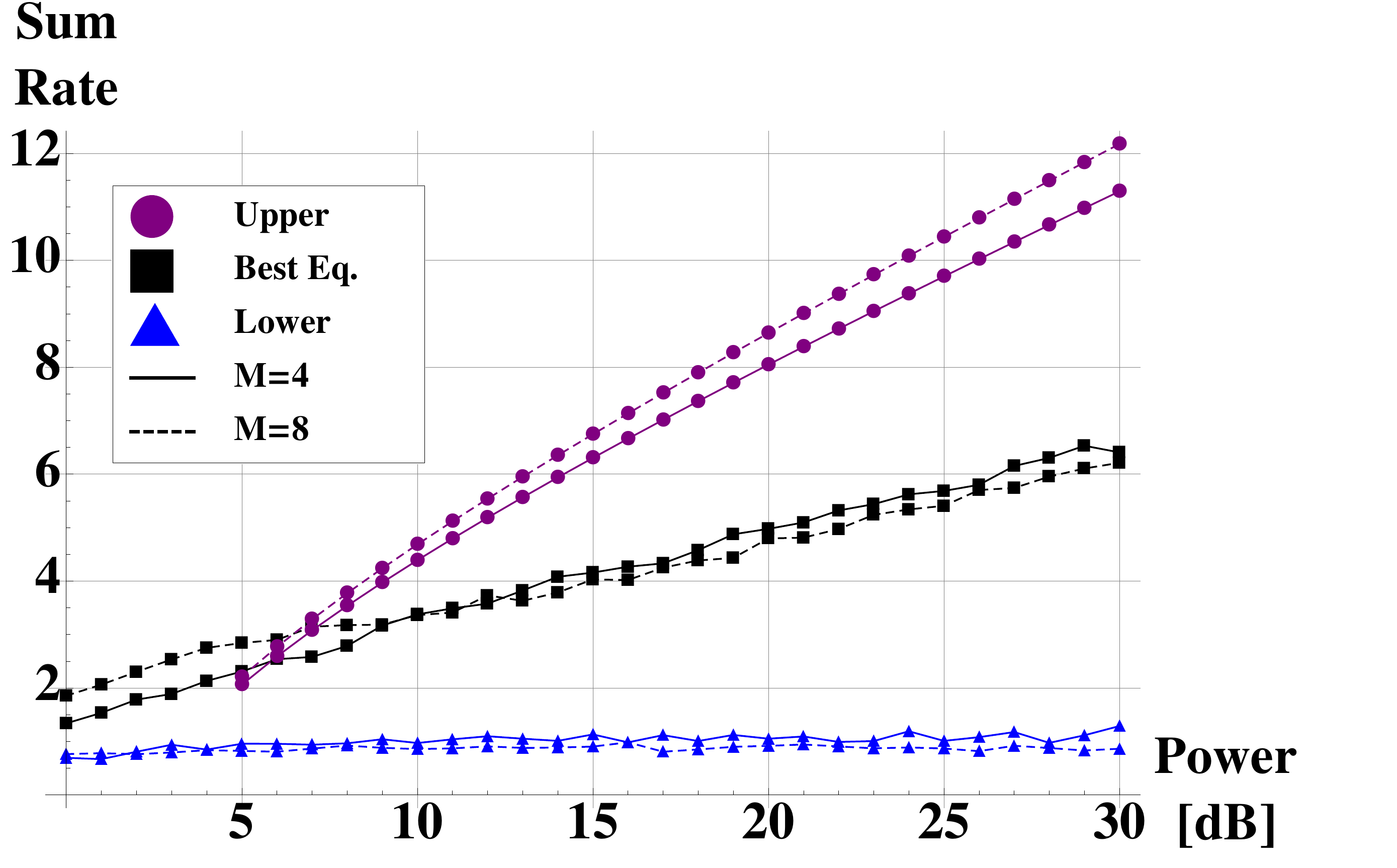}
    \caption{Simulation results for the average sum-rate compared with the upper and lower bounds given in \eqref{equ-Sum rate expression upper bounded} and \eqref{equ-Sum rate expression lower bounded}, respectively, as a function of the transmission power for $M=4$ and $M=8$.}
       \label{fig-SumRate Upper Lower As Func power}
\end{figure}

\section{Conclusion and future work}
This work gives evidence for the necessity of user scheduling under the CF scheme for large number of simultaneously transmitting users. We proved, with probability which goes to one, that in this regime the optimal choice of decoding at the relays is to decode the user with best channel instead of any other linear combination of the transmitted signals. Thus, CF becomes degenerate and it is more preferable to apply scheduling for much smaller group size. We show that even with a simple scheduling policy the sum-rate does not goes to zero and would like, as future work, to proceed and explore scheduling policies which exploit the decoding procedure of CF.

%\section*{Acknowledgment}

%The authors would like to thank Bobak Nazer for several fruitful discussions and enlightening comments.

% trigger a \newpage just before the given reference
% number - used to balance the columns on the last page
% adjust value as needed - may need to be readjusted if
% the document is modified later
%\IEEEtriggeratref{8}
% The "triggered" command can be changed if desired:
%\IEEEtriggercmd{\enlargethispage{-5in}}

% references section

% can use a bibliography generated by BibTeX as a .bbl file
% BibTeX documentation can be easily obtained at:
% http://www.ctan.org/tex-archive/biblio/bibtex/contrib/doc/
% The IEEEtran BibTeX style support page is at:
% http://www.michaelshell.org/tex/ieeetran/bibtex/
%\bibliographystyle{IEEEtran}
% argument is your BibTeX string definitions and bibliography database(s)
%\bibliography{IEEEabrv,../bib/paper}
%
% <OR> manually copy in the resultant .bbl file
% set second argument of \begin to the number of references
% (used to reserve space for the reference number labels box)

\appendices
\input{AppendixA}

%\input{Appendices/AppendixB}

%\onecolumn

%\input{Appendices/AppendixA}
%\input{Appendices/AppendixB}
%\input{Appendices/AppendixC}
%\input{Appendices/AppendixD}
%\twocolumn

%\section{Proof for Corollary \ref{cor-Probability for having a unit vector as the maximaizer goes to one}} \label{AppendixE} 

\bibliographystyle{IEEEtran}
\bibliography{../main_document/Bibliography}
%\afterpage{\blankpage}
%\includepdf[pages=-, offset=1 1]{Appendices/Appendices}

\end{document}

%% file: AppendixA.tex
% Appendix Template

\section{Proof for Theorem \ref{the-Sum rate is going to zero}} % Main appendix title

\label{AppendixA} % Change X to a consecutive letter; for referencing this appendix elsewhere, use \ref{AppendixX}

\begin{proof}
The probabilities $P(e_i)$ and $P(\overline{e_i})$ define a partition on the channel vectors a relays sees, specifically we define,
\begin{equation}\label{equ-Partition of channel vectors}
\begin{aligned}
& H_e=\{ \b{h}\in \R^L | \argmin_{\textbf{a}\in\mathbb{Z}^L \backslash \{\textbf{0}\}} f(\textbf{a})=e_i\} \\
& H_{\overline{e}}=\{ \b{h}\in \R^L | \argmin_{\textbf{a}\in\mathbb{Z}^L \backslash \{\textbf{0}\}} f(\textbf{a}) \neq e_i\}.
\end{aligned}
\end{equation}

That is, with probability $P(e_i)$ a relay sees a channel vector $\b{h}\in H_e$ and with probability $P(\overline{e_i})$ a relay sees a channel vector $\b{h}\in H_{\overline{e}}$. We note $\b{h}^{e}$ and $\b{h}^{\overline{e}}$ as a channel vectors which belongs to $H_e$ and $H_{\overline{e}}$ respectively.

Under the above definitions, the sum rate can be written as follows,
\begin{equation}\label{equ-The sum rate partitioned to two terms}
\begin{aligned}
 &\sum_{l=1}^{L} \min_{m:a_{ml} \neq 0} \cR(\b{h}_m,\b{a}_m)\\
 &= \sum_{l=1}^{L} P(e_i)\min_{m:e_{ml} \neq 0} \cR(\b{h}_m^{e},e_i) \\
 &\quad \quad \quad \quad \quad \quad \quad +P(\overline{e_i})\min_{m:a_{ml} \neq 0} \cR(\b{h}_m^{\overline{e}},\b{a}_m).
 \end{aligned}
\end{equation}

We treat the two terms above separately where the second term represents the sum rate for the case which the optimal coefficients vectors may be any integer vector excluding the unit vector $e_i$.  And the first term is for the case that the optimal coefficients vector is $e_i$.  We will show that both terms goes to zero while starting with the second term.

\begin{equation}
\begin{aligned}
 &\sum_{l=1}^{L} P(\overline{e_i})\min_{m:a_{ml} \neq 0} \cR(\b{h}_m^{\overline{e}},\b{a}_m)\\
 &=  P(\overline{e_i}) \sum_{l=1}^{L} \min_{m:a_{ml} \neq 0} \cR(\b{h}_m^{\overline{e}},\b{a}_m)\\
 & \leq P(\overline{e_i}) L \max_m \cR(\b{h}_m^{\overline{e}},\b{a}_m)\\
 &= P(\overline{e_i}) L \max_m \biggl\{ \\
 &\quad \left. \frac{1}{2}\log^+\left( \frac{1+P\|\mathbf{h}_m^{\overline{e}}\|^2}{ \|\mathbf{a}_m\|^2+ P(\|\mathbf{a}_m\|^2\|\mathbf{h}_m^{\overline{e}}\|^2 - ((\mathbf{h}_m^{\overline{e}})^T\mathbf{a}_m)^2)} \right)\right\}\\
 &\leq P(\overline{e_i}) L \max_m \frac{1}{2}\log^+\left( 1+P\|\mathbf{h}_m^{\overline{e}}\|^2 \right) \\
& = \max_m \left\{ P(\overline{e_i}) L \frac{1}{2}\log^+\left( 1+P\|\mathbf{h}_m^{\overline{e}}\|^2 \right) \right\}.
 \end{aligned}
\end{equation}

Define $R_L^{\overline{e}}= P(\overline{e_i}) L \frac{1}{2}\log^+\left( 1+P\|\mathbf{h}_m^{\overline{e}}\|^2 \right)$, we would like to show that $R_L^{\overline{e}} \overset{p}{\rightarrow} 0$, that is,
\begin{equation}
\lim_{L \rightarrow \infty}  P(R_L^{\overline{e}} \geq \epsilon)=0,
\end{equation}
for all $\epsilon > 0$.

Using The Markov and Jensen's inequalities we have, 
\begin{equation}\label{equ-Rate after Markov and Jensen's}
\begin{aligned}
&P(R_L^{\overline{e}} \geq \epsilon) \leq \frac{1}{\epsilon} \mathop{\mathbb{E}} \left[R_L^{\overline{e}} \right]\\
&=\frac{1}{\epsilon} \mathop{\mathbb{E}} \left[ P(\overline{e_i}) L \frac{1}{2}\log^+\left( 1+P\|\mathbf{h}_m^{\overline{e}}\|^2 \right) \right]\\
&\leq \frac{1}{\epsilon}  P(\overline{e_i}) \frac{L}{2}\log^+\left( 1+P \mathop{\mathbb{E}} \left[ \|\mathbf{h}_m^{\overline{e}}\|^2\right] \right),
\end{aligned}
\end{equation}

therefore, we are interested in analyzing the expectation of the squared norm values belonging to all channel vectors $\b{h}^{\overline{e}}$. 

Remember that, without any constraints, the channel vector $\b{h}$ is a Gaussian random vector which it's squared norm follows the $\chi^2_L$ distribution, we shell note as $f_{\chi^2_L}(x)$. A single squared norm value can belong to a several different Gaussian random vectors. Hence, we define $H_{\overline{e}}^{norm}$ as the set of squared norm values which belongs to $\b{h}^{\overline{e}}$, formally,

\begin{equation}
H_{\overline{e}}^{norm} =\{ \|\b{h}\|^2 \in \R | \b{h} \in H_{\overline{e}}\},
\end{equation}
which means in words, all possible squared norm values which belong to all vectors $\b{h}^{\overline{e}}$. We define then, $P(\xi)=P(\|\b{h}\|^2 \in H_{\overline{e}}^{norm})$ as the probability to belong to $H_{\overline{e}}^{norm}$.  That is,

\begin{equation}
\int \limits_{H_{\overline{e}}^{norm}} f_{\chi^2_L}(x) dx = P(\xi).
\end{equation}

Returning to the expectation in \eqref{equ-Rate after Markov and Jensen's} we have,

\begin{equation}
\begin{aligned}
\mathop{\mathbb{E}} \left[ \|\mathbf{h}_m^{\overline{e}}\|^2\right] &= \int \limits_{H_{\overline{e}}^{norm}} x \frac{f_{\chi^2_L}(x)}{P(\xi)} dx 
\leq  \int_{\alpha}^{\infty} x \frac{f_{\chi^2_L}(x)}{P(\xi)} dx\\
& \overset{(a)}{\leq}  \int_{\alpha}^{\infty} x \frac{f_{\chi^2_L}(x)}{P(\overline{e_i})} dx
\leq \frac{1}{P(\overline{e_i})} \int_{0}^{\infty} x f_{\chi^2_L}(x) dx\\
&= \frac{L}{P(\overline{e_i})}. 
\end{aligned}
\end{equation}
Where, $\alpha$ satisfies $\int_{\alpha}^{\infty} f_{\chi^2_L}(x) dx = P(\xi)$ and $(a)$ is due to the fact that $P(\xi) \geq P(\overline{e_i})$ since it may happen that two vectors $\b{h}^{e}$ and $\b{h}^{\overline{e}}$ would have the same squared norm value.

Applying the expectation's upper bound in \eqref{equ-Rate after Markov and Jensen's} we have,

\begin{equation}\label{equ-Rate after Markov and Jensen's and upper bound on expectation}
\begin{aligned}
& \frac{1}{\epsilon}  P(\overline{e_i}) \frac{L}{2}\log^+\left( 1+P \mathop{\mathbb{E}} \left[ \|\mathbf{h}_m^{\overline{e}}\|^2\right] \right)\\
&  \leq \frac{1}{\epsilon}  P(\overline{e_i}) \frac{L}{2} \log^+\left( 1+\frac{PL}{P(\overline{e_i})}  \right)\\
&\overset{(a)}{\leq} \frac{1}{\epsilon}  P(\overline{e_i}) \frac{L}{2} \sqrt{2 \frac{PL}{P(\overline{e_i})} }
= \frac{1}{\epsilon} \frac{L}{2} \sqrt{2PLP(\overline{e_i})}  \\
&\overset{(b)}{\leq} \frac{1}{\epsilon} \frac{L}{2}   \sqrt{8P^2L^3e^{-LE_2(L)}}
= \frac{1}{\epsilon} PL^2\sqrt{2L}e^{-LE_3(L)}
\end{aligned}
\end{equation}  

where $(a)$ is due the bound $\log(1+x)\leq \sqrt{2x}$, $(b)$ following directly from Theorem \ref{the-Probability for having a unit vector as the maximaizer over all other vectors} and $E_3(L)=\frac{1}{4}(1-\frac{1}{L})\log{2}$.
Considering the above, as L grows, the second term of \eqref{equ-The sum rate partitioned to two terms} is going to zero, that is,

\begin{equation}
\begin{aligned}
&\lim_{L \rightarrow \infty}  \sum_{l=1}^{L} P(\overline{e_i})\min_{m:a_{ml} \neq 0} \cR(\b{h}_m^{\overline{e}},\b{a}_m) \\
&\leq \lim_{L \rightarrow \infty} \frac{1}{\epsilon} \frac{PL^2\sqrt{L}}{\sqrt{0.5}}e^{-LE_3(L)} =0
\end{aligned}
\end{equation}
for all $\epsilon>0$.

Thus, we are left with the first term in \eqref{equ-The sum rate partitioned to two terms}
\begin{equation}
\begin{aligned}
 &\lim_{L \rightarrow \infty} \sum_{l=1}^{L} P(e_i) \min_{m:e_{il} \neq 0} \cR(\b{h}_m,e_i)
 \leq \sum_{l=1}^{L} \min_{m:e_{il} \neq 0} \cR(\b{h}_m,e_i) \\
& \overset{(a)}{=} \lim_{L \rightarrow \infty} \sum_{l=1}^{L} \min_{m:e_{il} \neq 0} \frac{1}{2} \log^+\left( \frac{1+P\|\b{h}_m\|^2}{1+P(\|\b{h}_m\|^2-h_i^2)} \right)\\
 & \overset{(b)}{\leq} \lim_{L \rightarrow \infty} \sum_{m=1}^{M} \frac{1}{2} \log^+\left( \frac{1+P\|\b{h}_m\|^2}{1+P(\|\b{h}_m\|^2-h_i^2)} \right)\\
& = \sum_{m=1}^{M} \frac{1}{2} \log^+\left(  \lim_{L \rightarrow \infty} \frac{1+P\|\b{h}_m\|^2}{1+P(\|\b{h}_m\|^2-h_i^2)} \right) =0
 \end{aligned}
\end{equation}
where in $(a)$ we set the unit vector $e_i$ in the rate expression $\cR(\b{h}_m,\b{a}_m)$. The upper bound $(b)$ is for the best case scenario for which each relay has different unit vector $e_i$. Finally, it is clear that as $L$ grows for each realization of $\b{h}_m$ the argument of the $\log$ is going to $1$.
\end{proof}